\documentclass[12pt,reqno, letterpaper]{amsart}
\usepackage{amssymb,epsf}
\usepackage{amstext,amsmath,amsfonts,amscd,amsthm,graphicx}
\usepackage{epsfig}
\usepackage{eufrak}
\usepackage{latexsym}
\usepackage{eucal}
\theoremstyle{amsart}
\newfont{\fnt}{cmsy10}
\newfont{\sss}{cmr10}
\newfont{\azb}{wncyr10}
\newfont{\azbit}{wncyi10}
\theoremstyle{definition}

\theoremstyle{plain}

\newtheorem{vt}{Theorem}

\newtheorem{prp}{Proposition}
\theoremstyle{definition}

\begin{document}
\title[Complete list of conservation laws for non-integrable compacton equations
%of
%$K(m,m)$ type
]{
{\protect\vspace*{-2.3cm}}A complete list of conservation laws for non-integrable compacton equations
of $K(m,m)$ type}
\author{Ji\v{r}ina Vodov\'a}
\address{Mathematical Institute, Silesian University in Opava, Na Rybn\'{i}\v{c}ku \nolinebreak 1, 746 01 Opava, Czech Republic}
\email{Jirina.Vodova@math.slu.cz}
\maketitle
\begin{abstract} {\protect\vspace*{-0.7cm}}
In 1993, P. Rosenau and J. M. Hyman introduced and studied Korteweg-de-Vries-like equations with nonlinear dispersion admitting compacton solutions, $u_t+D_x^3(u^n)+D_x(u^m)=0$, $m,n>1$, which are known as the $K(m,n)$ equations.
In the present paper we consider a slightly generalized version of the $K(m,n)$ equations for $m=n$, namely,
$u_t=aD_x^3(u^m)+bD_x(u^m)$, where $m,a,b$ are arbitrary real numbers. We describe all generalized symmetries and conservation laws thereof
for $m\neq -2,-1/2,0,1$; for these four exceptional values of $m$ the equation in question is either completely integrable ($m=-2,-1/2$) or linear ($m=0, 1$). It turns out that for $m\neq -2,-1/2,0,1$ there are only three symmetries corresponding to $x$- and $t$-translations and scaling of $t$ and $u$,
and four nontrivial conservation laws, one of which expresses the conservation of energy, and the other three are associated with the Casimir functionals of the Hamiltonian operator $\mathfrak{D}=aD_x^3+bD_x$ admitted by our equation. Our result provides {\em inter alia} a rigorous proof of the fact that the $K(2,2)$ equation has just four conservation laws found by P. Rosenau and J. M. Hyman.\looseness=-1
%, and we find %all conservation laws for the generalized $K(m,m)$ equations, where  $m$ is an arbitrary real number such that the corresponding $K(m,m)$ equation is %not integrable. These results were proved using a novel method involving a nonstandard use of a Hamiltonian operator in conjunction with the formal %symmetry approach.

\looseness=-1
\end{abstract}
\section*{Introduction}
The equations possessing soliton solutions with compact support (the so-called \textit{compactons})
are presently of great interest for both mathematicians and physicists, see e.g.\ \cite{OR,P+R,  Rosenau, Rosenau2, compactons},
because such equations can provide adequate models for natural phenomena with a finite span.
Initially compactons emerged as solutions of fully nonlinear Korteweg-de-Vries-like equations (the $K(m,n)$ equations):
$$u_t+D_x^3(u^n)+D_x(u^m)=0,$$
which have first appeared in \cite{compactons}; here $D_x$ denotes the total $x$-derivative
\[
D_x=\displaystyle\frac{\partial}{\partial x}+\sum\limits_{j=0}^\infty u_{j+1}\frac{\partial}{\partial u_j},
\]
$m,n>1$, $t$ is the time and $x$ is the space variable, $u_j$ denotes $j$th derivative of $u$ with respect to $x$, $u_0\equiv u$, see e.g.\ \cite{wh_integrability, Olver, symmetries} for further details on this notation.
% These equations have a remarkable property: their solitary wave solutions collide elastically, but unlike the Korteweg-de-Vries ($m=2,n=1$) solitons, they have compact support (\cite{compactons}).

Although the solitary waves have compact support only if $n>1$ and a compacton is a solution for a $K(m,n)$ equation in the classical sense only for $n\leq 3$ \cite{compactons}, it is natural to study a slightly more general version of these equations which we hereinafter refer to as {\em generalized $K(m,n)$ equations}:
\begin{equation}\label{gkmn}
u_t=aD_x^3(u^n)+bD_x(u^m),
\end{equation}
where $a,b,m,n$ are arbitrary real numbers.

If $m=n$, these equations are easily seen to be Hamiltonian with respect to the Hamiltonian operator $\mathfrak{D}=aD_x^3+bD_x$, the Hamiltonian functional being $\mathcal{H}=\int\int u^m d u d x$. Thus,  equation (\ref{gkmn}) for $m=n$ can be written as
\begin{equation}\label{gkmm}
u_t=aD_x^3(u^m)+bD_x(u^m)=\mathfrak{D}\delta \mathcal{H},
\end{equation}
where $\delta$ denotes the variational derivative of a functional with respect to $u$ and
$\int dx$ is understood as a formal integral in the sense of calculus of variations, see e.g.\ \cite{dickey, Olver} for details.
Recall (see e.g.\ \cite{Olver}) that for any functional $\mathcal{F}=\int f(x,t,u,\dots,u_s) dx$ with a smooth density $f$ we have
\[
\delta\mathcal{F}=\displaystyle\frac{\delta f}{\delta u}\,,\ \ \mbox{\rm where}\ \
\displaystyle\frac{\delta}{\delta u}=\sum\limits_{j=0}^\infty (-D_x)^j \displaystyle\frac{\partial}{\partial u_j}.
\]

The pseudo-differential operator $\mathfrak{D}^{-1}$ is easily seen to be a formal conservation law of rank $\infty$ for (\ref{gkmm}) in the sense of \cite{Olver}. This means \cite{wh_integrability, Olver}
that an infinite set of ``standard" obstacles for existence of infinitely many conservation laws of increasing order for (\ref{gkmm}) vanishes,
and therefore one could expect that this equation should share at least some of the properties of integrable PDEs.
However, from the results of \cite{wh_integrability} it can be inferred that only the  equations
%$K(-1,-2)$,
 $K(-2,-2)$ and $K(-\frac{1}{2},-\frac{1}{2})$
% and $K(\frac{3}{2},-\frac{1}{2})$
are integrable (cf.\ \cite{Rosenau} for an alternative argument and Proposition~\ref{prop1} below).

Thus, there are no other nonlinear symmetry integrable equations among those of the generalized $K(m,m)$ type, i.e., of the form (\ref{gkmm}).
Furthermore, we were able to obtain a complete
description of generalized symmetries, including those explicitly dependent on $t$ and $x$, for (\ref{gkmm}) with $m\neq -2,-1/2,0,1$, see Proposition~\ref{prop2} below and subsequent discussion.

Knowing these in conjunction with the Hamiltonian operator $\mathfrak{D}$ has allowed us to provide in Theorem~\ref{2} below a complete description of the conservation laws for (\ref{gkmm}) with $m\neq -2,-1/2,0,1$. It turns out that, apart from the energy, all conserved functionals for our generalized $K(m,m)$ equation are Casimir functionals for the Hamiltonian operator $\mathfrak{D}$, and therefore they are \cite{Olver} conserved functionals for any evolution equation of the form
$u_t=\mathfrak{D}f$, where $f$ is an arbitrary smooth % or e.g.\ meromorphic
function of $x,u$ and a finite number of $u_j$.

Note that in \cite{compactons} four conservation laws for the $K(2,2)$ equation were found and it was claimed that no other exist. To the best of our knowledge this claim hasn't yet been proved. Our results provide {\em inter alia} a rigorous proof of this assertion.
%but also give the complete list of conservation laws for a more general form of $K(m,m)$ equations for
%$m\neq -2,-1/2,0,1$.

Recall that the conservation laws could be employed e.g.\ for the study of stability and for the proof of existence and uniqueness of the solution which belongs to a given function space. Knowing a complete list of conservation laws for a given equation is required, for instance, in order to find
normal forms thereof with respect to low-order conservation laws \cite{ps}, and
for constructing higher-precision discretizations, because it is desirable that the latter preserve all conservation laws of the equation in question, cf.\ e.g.\ \cite{lw} and references therein.

\section{Preliminaries}

Recall (cf.\ e.g.\ \cite{vin, wh_integrability, Olver}) that a (smooth)
function $f$ is called {\em local} if it depends only on $x,t,u$ and a finite number of $u_j$.

Consider an evolution equation with the local right-hand side of the form
\begin{equation}\label{ee}
u_t=F(x, u,u_1,\dots,u_k),\quad k\geq 2.
\end{equation}
%where $F$ is local.

A local function $G=G(x,t,u,u_1,\dots,u_s)$
is called (see e.g.\ Chapter 5 of \cite{Olver} for details)
a {\em characteristic of generalized symmetry} for (\ref{ee}) if it satisfies the linearized version of the latter, that is,
\[
D_t(G)=\mathrm{D}_F(G),
\]
where $D_t=\partial/\partial t+ \sum\limits_{j=0}^\infty D_x^j (F)\partial/\partial u_j$ is the total $t$-derivative,
and $\mathrm{D}_F=\sum\limits_{j=0}^k \partial F/\partial u_j D_x^j$.

Next, a {\em formal symmetry of order} $q$ for (\ref{ee}) is \cite{wh_integrability, symmetries} a formal series of the form
\[
\mathfrak{L}=\sum\limits_{j=-\infty}^1 a_j D_x^j,
\]
where $a_j$ are local functions,
such that
\[\mathrm{deg}(
D_t (\mathfrak{L})-[\mathrm{D}_F,\mathfrak{L}]
)\leq k+1-q.\]
Here the symbol $\mathrm{deg}$ stands for the degree of a formal series; recall that
for $\mathfrak{M}=\sum\limits_{j=-\infty}^m b_j D_x^j$ with $b_m\neq 0$
we have $\mathrm{deg}\, \mathfrak{M}=m$ by definition, cf.\ e.g.\ \cite{wh_integrability, symmetries}.
%The reader is again referred to \cite{Olver} for further details.

Finally, (\ref{ee})  is called \textit{symmetry integrable} \cite{wh_integrability, Olver, symmetries}
if it admits an infinite sequence of explicitly time-independent generalized symmetries of increasing order.

To prove the claims made in Introduction, we shall first need the following result based on the symmetry approach to integrability, see e.g.\ \cite{symmetries,wh_integrability,complete_lists,test} for details:
\begin{vt}[\cite{wh_integrability,complete_lists}]\label{1}
An equation (\ref{ee}) possesses an explicitly time-independent formal symmetry of order $N>k$
if and only if the first $N-k$ canonical densities $\rho_i$, $i=-1,0,1,2,\dots,N-k-2$,
are densities of local conservation laws.

Existence of an explicitly time-independent formal symmetry of order $q>N$ is a necessary condition for (\ref{ee}) to possess explicitly time-independent generalized
symmetries with the characteristic of order $q$.
\end{vt}

Hence, existence of an explicitly
time-independent formal symmetry of infinite order is a necessary condition for (\ref{ee})
to be symmetry integrable.

As for the canonical densities for (\ref{ee}), these can be computed recursively from a formal symmetry $\mathfrak{L}$
of sufficiently high order; explicit formulas for a few of them can be found
in \cite{symmetries,wh_integrability,complete_lists,test}. For instance, for $k>2$ we have
\begin{equation}\label{ccd}
\rho_{-1}=(\partial F/\partial u_k)^{-1/k}\ \mbox{and} \ \rho_0=(\partial F/\partial u_{k-1})/(\partial F/\partial u_k).
\end{equation}

\section{Main results}

\begin{prp}\label{prop1}
 If $m\neq -2,-1/2,0,1$, then the corresponding generalized $K(m,m)$ equation (\ref{gkmm}) has
 no {\em explicitly time-independent} generalized symmetries of order greater than 3; in particular, equation (\ref{gkmm})
 is not symmetry integrable.
\end{prp}
\begin{proof}
Applying (\ref{ccd}) %for the first canonical density of (\ref{ee})
to (\ref{gkmm}) we obtain that $\rho_{-1}= (a m u^{m-1})^{-1/3}$.
Computing the quantity $\delta D_t(\rho_{-1})/\delta u$ reveals that it is identically equal to zero only for these values of $m$, i.e., $m=-2,-1/2,0,1$. Hence (cf.\ e.g.\ \cite{Olver, symmetries})
our $\rho_{-1}$ is a density for a conservation law of our equation only for $m=-2,-1/2,0,1$. By Theorem~\ref{1}, our equation (\ref{gkmm}) for $m\neq -2,-1/2,0,1$ cannot have an explicitly time-independent formal symmetry of order greater than 3, and therefore it cannot possess any explicitly time-independent generalized symmetries of order greater than 3 and, in particular, it cannot be symmetry integrable.\looseness=-1
\end{proof}
%Note that it is also possible (although that would be somewhat more cumbersome) to prove the above result by
%transforming (\ref{gkmm}) into an equation with the constant separant and then make direct use of the classification results
%from \cite{wh_integrability} where all such third-order equations are listed.

Proposition~\ref{prop1} ensures non-existence of {\em explicitly time-independent} generalized symmetries of order greater than 3.
However, this result can be further generalized to explicitly time-dependent symmetries:
\begin{prp}\label{prop2}
 If $m\neq -2,-1/2,0,1$, then the corresponding generalized $K(m,m)$ equation (\ref{gkmm}) has
 no generalized symmetries, including explicitly time-dependent ones, of order greater than 3.

The only generalized symmetries of (\ref{gkmm}) for $m\neq -2,-1/2,0,1$ are those with the characteristics
 $Q_1=u_x$, $Q_2=u_t$ and $Q_3=(m-1)t u_t+u$, i.e., $x$- and $t$-translations and the scaling symmetry.
\end{prp}
\begin{proof}
Indeed, for $m\neq -2,-1/2,0,1$ the above $\rho_{-1}$ is not a conserved density for (\ref{gkmm}).
In conjunction with Theorem~2 of \cite{as} this is readily seen to imply that (\ref{gkmm}) has
no formal symmetry, even explicitly time-dependent one, of order greater than 3. It is now readily
seen (cf.\ also Theorem~1 in \cite{as}) that (\ref{ee}) has no generalized symmetries, including explicitly time-dependent ones,
of order greater than 3.

Now that we know that all generalized symmetries of (\ref{gkmm}) with
$m\neq -2,-1/2,0,1$ are of order at most 3, we can readily find all of them.
%If we invoke Proposition \ref{prop1} we can find all symmetries of the $K(m,m)$ equations, where $m\neq -2,-1/2,0,1$.
%A straightforward computation shows
It turns out that for $m\neq -2,-1/2,0,1$  equation~(\ref{gkmm}) has
three symmetries with the characteristics $Q_1=u_x$, $Q_2=u_t$ and $Q_3=(m-1)tu_t+u$.
\end{proof}

Note that there are no conservation laws associated (through the Hamiltonian operator $\mathfrak{D}$)
to the first and third symmetry. The conserved functional associated to the second symmetry through $\mathfrak{D}$
is the energy $\int \int  u^m du dx$.
\begin{vt}\label{2}
If $\rho$ is a density of a local conservation law for a generalized $K(m,m)$ equation, where $m\neq -2,-1/2,0,1$, then it is, up to the addition of a trivial density, a function of $x,t$ and $u$ only.
\end{vt}
\begin{proof}
Let $\rho$ be a density of a local conservation law for a non-integrable case of the $K(m,m)$ equation, i.e., for $m\neq -2,-1/2,0,1$. Then the function $\gamma=\frac{\delta\rho}{\delta u}$ is a cosymmetry for (\ref{gkmm}).
As the $K(m,m)$ equation (\ref{gkmm}) is Hamiltonian with respect to the Hamiltonian operator $\mathfrak{D}=aD_x^3+bD_x$, see (\ref{gkmm}),
and thanks to the fact that Hamiltonian operators map cosymmetries to symmetries, see e.g.\ \cite{blaszak}, we conclude that
$\mathfrak{D}(\gamma)$ is a symmetry of our equation (\ref{gkmm}). But it follows
from Proposition~\ref{prop2} that our equation cannot have
generalized symmetries of order greater than 3, so the order of $\mathfrak{D}(\gamma)$ is less than or equal to $3$. Now suppose that the order of the function $\gamma$ is equal to $K>0$. Then the order of $\mathfrak{D}(\gamma)$ would be equal to $3+K>3$, which is a contradiction, so the order
of any cosymmetry $\gamma$ of (\ref{gkmm}) with $m\neq -2,-1/2,0,1$ must be zero, i.e., it may depend at most on $x,t,u$. To any such cosymmetry $\gamma=\gamma(x,t,u)$ there corresponds a conserved density $\rho=\int \gamma du$ for which $\delta\rho/\delta u=\gamma$ by construction. Quite obviously, this $\rho$ also depends at most on $x,t,u$ and is defined up to the addition of an arbitrary (smooth) function of $x$ and $t$.\looseness=-1
\end{proof}
Knowing from Theorem \ref{2} that up to the addition of a trivial density any conserved density for (\ref{gkmm}) with $m\neq -2,-1/2,0,1$ depends at most on $x,t,u$,
we can readily find all conservation laws of all $K(m,m)$ equations for $m\neq -2,-1/2,0,1$. We omit the straightforward computations and state just the relevant result:
\begin{vt}
The only local conservation laws of the form $D_t(\rho)=D_x(\sigma)$ for the generalized $K(m,m)$ equation (\ref{gkmm}) with $m\neq -2,-1/2,0,1$,
are, modulo the addition of trivial conservation laws, just the linear combinations of the four conservation laws
which for $b\neq 0 $ are given by the formulas
$$
\begin{array}{ll}
\rho_1=\displaystyle\int u^m d u&\sigma_1=\displaystyle\left(m a u_{xx}u^{2m-1}+\frac{a m(m-2)}{2}u_x^2u^{2m-2}+\frac{b}{2}u^{2m}\right)\\[5mm]
\rho_2=u&\sigma_2=aD_x^2(u^m)+bu^m\\[5mm]
\displaystyle\rho_3=u\sin\left(\frac{\sqrt{b}}{\sqrt{a}}x\right)&\sigma_3=aD_x^2(u^m)\sin\left(\frac{\sqrt{b}}{\sqrt{a}}x\right)
-\sqrt{ab}D_x(u^m)\cos\left(\frac{\sqrt{b}}{\sqrt{a}}x\right)\\[5mm]
\displaystyle\rho_4=u\cos\left(\frac{\sqrt{b}}{\sqrt{a}}x\right)&\sigma_4=aD_x^2(u^m)\cos\left(\frac{\sqrt{b}}{\sqrt{a}}x\right)
+\sqrt{ab}D_x(u^m)\sin\left(\frac{\sqrt{b}}{\sqrt{a}}x\right).
\end{array}
$$
If $b=0$, then the conservation law with the density $\rho_3$ is trivial, and the densities $\rho_2$ and $\rho_4$ coalesce.
However, there are two other conservation laws in such a case, namely
$$
\begin{array}{ll}
\rho_5=xu&\sigma_5=aD_x^2(xu^m)-3aD_x(u^m)\\
\rho_6=x^2u&\sigma_6=aD_x^2(x^2u^m)+6au^m-aD_x(xu^m),
\end{array}
$$
i.e., for $b=0$ equation (\ref{gkmm}) with $m\neq -2,-1/2,0,1$ also has,
up to the addition of trivial conservation laws, just four conservation laws
with the densities $\rho_1,\rho_2,\rho_5,\rho_6$ and the fluxes $\sigma_1,\sigma_2,\sigma_5,\sigma_6$.
\end{vt}

Note that if $a$ and $b$ have different signs then sines and cosines of a complex variable appear in the formulas for $\rho_3,\rho_4,\sigma_3$ and $\sigma_4$. In this case it is convenient to divide $\rho_3$ by the imaginary unit $\mathrm{i}$ and use the following {\em real} densities and fluxes instead of the above $\rho_3,\rho_4,\sigma_3$ and $\sigma_4$:
$$
\begin{array}{ll}
\tilde{\rho}_3=\displaystyle cu\sinh\left(\frac{\sqrt{|b|}}{\sqrt{|a|}}x\right)&
\tilde{\sigma}_3=\displaystyle caD_x^2(u^m)\sinh\left(\frac{\sqrt{|b|}}{\sqrt{|a|}}x\right)
-\sqrt{|ab|}D_x(u^m)\cosh\left(\frac{\sqrt{|b|}}{\sqrt{|a|}}x\right)\\[5mm]
\tilde{\rho}_4=\displaystyle u\cosh\left(\frac{\sqrt{|b|}}{\sqrt{|a|}}x\right)&
\tilde{\sigma}_4=\displaystyle aD_x^2(u^m)\cosh\left(\frac{\sqrt{|b|}}{\sqrt{|a|}}x\right)
-c \sqrt{|ab|}D_x(u^m)\sinh\left(\frac{\sqrt{|b|}}{\sqrt{|a|}}x\right),
\end{array}
$$
where $c=1$ if $a>0$ and $b<0$, and $c=-1$ if $a<0$ and $b>0$.

The conserved functional corresponding to the first conserved density is the energy, i.e., the integral of motion associated with the invariance under the time shifts. If $m=2k-1$ where $k\in\mathbb{Z}\setminus \left\{0,1\right\}$, then fact that the quantity $\int u^{m+1} dx$ is conserved immediately implies the following property of the solutions of the corresponding $K(m,m)$ equation:  if a solution $u(x,t)$ of (\ref{gkmm}) belongs to the space $L^{2k}(\mathbb{R})$ as a function of $x$, i.e., $\int_\mathbb{R} |u|^{2k} dx <\infty$, at the time $t=t_0$ then $u(x,t)\in L^{2k}(\mathbb{R})$ for all $t\geq t_0$.

The remaining conserved functionals are Casimir functionals  corresponding to our Hamiltonian operator $\mathfrak{D}$,
so finding a suitable physical interpretation thereof is rather unlikely.

As a final remark, note that it would be interesting to apply our method for proving nonexistence of higher conservation laws using
existence of a Hamiltonian operator to other nonintegrable systems.

\section*{Acknowledgements}
The author thanks Dr. A. Sergyeyev for stimulating discussions. This research was supported by the Silesian university in Opava under the student grant SGS/11/2012, by the fellowship
from the Moravian--Silesian region, and by the the Ministry of Education, Youth and Sport of the Czech Republic
through  institutional support for the development of research organization (I\v C47813059).\looseness=-1

\end{document}